\documentclass[11pt]{article}
\usepackage{amssymb, amsmath, amsthm}
\usepackage{xspace}
\usepackage{calc}
\usepackage[sort&compress,numbers]{natbib}
\usepackage{enumitem}
\usepackage{algorithm}
\usepackage{times}
\usepackage[T1]{fontenc}

\usepackage{amssymb}        
\usepackage{latexsym}
\usepackage{algpseudocode}
\usepackage{graphicx}
\usepackage{verbatim}

\usepackage{tikz}
\usepackage{subcaption}
\usetikzlibrary{calc}
\usetikzlibrary{decorations.markings,decorations.pathmorphing,positioning,intersections}

\newtheorem{theorem}{Theorem}

\newtheorem{definition}{Definition}

\newtheorem{claim}{Claim}

\newcommand{\lbmatching}{$(l,u)$-matching}

\usepackage[textwidth=20mm,disable]{todonotes}
\newcommand{\mytodo}[2]{\todo[size=\tiny, color=#1!50!white]{#2}}

\newcommand{\mwcom}[1]{\mytodo{green}{#1}}

\usepackage{hyperref}

\author{Katarzyna Paluch \thanks{\texttt{abraka@cs.uni.wroc.pl}} \hspace{0.1cm}
   and Mateusz Wasylkiewicz \thanks{\texttt{mateusz.wasylkiewicz@cs.uni.wroc.pl}}\\
	Institute of Computer Science,  University of Wroc{\l}aw \\}

\title{A   simple combinatorial algorithm for restricted \mwcom{maybe add "weighted" to title?} 2-matchings  in subcubic graphs - via half-edges\thanks{Partially supported by Polish National Science Center grant 2018/29/B/ST6/02633.}}

\date{}

\begin{document}
\maketitle

\begin{abstract}
We consider three variants of the  problem of finding a maximum weight restricted $2$-matching  in a subcubic graph $G$. (A $2$-matching is any subset of the edges such that each vertex is incident to at most two of its edges.) Depending on the variant a {\em restricted} $2$-matching means a $2$-matching that is either triangle-free or square-free or both triangle- and square-free. While there exist polynomial time algorithms for the first two types of  $2$-matchings%solving it
, they are quite complicated or use advanced methodology.  For each of the three problems we present a simple reduction  to the computation of a maximum weight $b$-matching. The reduction is conducted with the aid of {\em half-edges}.  A {\em half-edge} of edge $e$ is, informally speaking, a half  of $e$
containing exactly one of its endpoints. For a subset of triangles of $G$, we replace each edge  of such a  triangle with two half-edges. Two half-edges of one edge $e$ of weight $w(e)$ may get different weights, not necessarily equal to $\frac{1}{2}w(e)$.   In the metric setting when the edge weights satisfy the triangle inequality, this  has a geometric interpretation connected to how an  incircle partitions the edges of a triangle. 
Our algorithms are additionally faster than those known before. The  running time of each of them is $O(n^2\log{n})$, where $n$ denotes the number of vertices in the graph.
\end{abstract}

\tikzset{vertex/.style={minimum size=#1,circle,fill=black,draw,inner sep=0pt},
	decoration={markings,mark=at position .5 with {\arrow[black,thick]{stealth}}},
	vertex/.default=2.5mm,
	bigVertex/.style={vertex=4mm},
	e0/.style={line width=0.8pt},
	e1/.style={line width=0.8mm},
	transformsTo/.pic=
	{
		\coordinate (-leftEnd) at (0,0);
		\coordinate (-rightEnd) at (5,0);
		\draw[thick] (0,0) -- (0,1) -- (3,1) -- (3,2) -- (5,0) -- (3,-2) -- (3,-1) -- (0,-1) -- (0,0);
	}}

\section{Introduction}
A subset~$M$ of edges of an undirected simple graph is a \emph{2-matching} if every vertex is incident to at most two  edges of~$M$. $2$-matchings belong to a wider class of $b$-matchings, where for every vertex $v$ in the set of vertices $V$ of the graph, we are given a natural number $b(v)$ and a subset of edges is a \emph{$b$-matching} if every vertex is incident to at most $b(v)$ of its  edges.  A $2$-matching is called \emph{$C_k$-free} if it does not contain any cycle of length at most~$k$. Note that every $2$-matching is $C_2$-free and the smallest length of a cycle in a $2$-matching is three.  A $2$-matching of maximum size can be found in polynomial time by a reduction to a classical matching. The \emph{$C_k$-free $2$-matching problem} consists in finding a $C_k$-free $2$-matching of maximum size. Observe that the $C_k$-free $2$-matching problem for $n/2 \leq k < n$, where $n$ is the number of vertices in the graph, is equivalent to finding a Hamiltonian cycle, and thus NP-hard.
Hartvigsen \cite{Hartvigsen1984} gave a complicated algorithm for the case of  $k=3$. Papadimitriou \cite{CornuejolsPulleyblank1980} showed that this problem is NP-hard when $k\geq 5$. The complexity of the $C_4$-free $2$-matching problem is unknown.

In the weighted version of the problem,  each edge $e$ is associated with a  nonnegative weight $w(e)$ and we are interested in finding a $C_k$-free $2$-matching of maximum  weight, where the weight of a $2$-matching $M$ is defined as the sum of weights of edges belonging to $M$.  %The case $k=2$ is classical weighted $2$-matching problem which can be solved in polynomial time\mwcom{add sources}. 
Vornberger~\cite{Vornberger1980} showed that the weighted $C_4$-free $2$-matching problem is NP-hard. % Finding complexity of $C_3$-free $2$-matching problem is open problem. 
We refer to cycles of length three and four as \emph{triangles} and \emph{squares}, respectively.

In the paper we consider the following three problems in subcubic graphs:
the weighted triangle-free $2$-matching problem (i.e. the weighted $C_3$-free $2$-matching problem), the weighted square-free $2$-matching problem, in which we want to find a maximum weight $2$-matching without any squares, but possibly containing triangles and the weighted $C_4$-free $2$-matching problem. A graph is called {\em cubic} if its every vertex has degree $3$ and is   called {\em subcubic} if its every vertex has degree at most $3$. 

{\em The weighted triangle-free $2$-matching problem in subcubic graphs.} The existing two polynomial time algorithms for this problem are the following. 
 Hartvigsen and Li~\cite{HartvigsenLi2012} gave %complete description of the convex hull of incidence vectors of triangle-free $2$-matchings for subcubic graphs. It uses a type of so-called \emph{comb inequality}. They gave also polynomial time
a rather complicated primal-dual algorithm with running time $O(n^3)$ and a long analysis. The algorithm uses a type of so-called comb inequality.   Kobayashi~\cite{Kobayashi2010} devised a simpler algorithm using the theory of $M$-concave functions on finite constant-parity jump systems as well as makes $O(n^3)$ computations of a maximum weight $b$-matching for $b\in\{0,1,2\}^V$. Its running time is $O(n^5\log{n})$.

We present a simple combinatorial algorithm for the problem that uses one computation of a maximum weight $b$-matching for $b\in\{0,1,2\}^V$.
Given a subcubic graph $G$, we replace some  of its triangles with gadgets containing {\em half-edges} and define a function $b$ on the set of vertices in such a way that, any $b$-matching in the thus constructed graph $G'$ yields a triangle-free $2$-matching. 
A {\em half-edge} of edge $e$ is, informally speaking, a half  of $e$
containing exactly one of its endpoints. Half-edges have already been  introduced in \cite{PaluchEtAl2012} and used in several subsequent papers. Here we use a different weight distribution
among half-edges of one edge -
  two half-edges of one edge $e$  may be assigned different  weights and not necessarily equal to $\frac{1}{2}w(e)$.    In the metric setting when the edge weights satisfy the triangle inequality, this  has a geometric interpretation connected to how an  incircle partitions the edges of a triangle. The running time of our algorithm is $O(n^2\log{n})$. If the graph is unweighted, then the run time of this algorithm
	becomes $O(n^{3/2})$.

{\em Square-free $2$-matchings.} In bipartite graphs a shortest cycle has length four - a square.  Polynomial time algorithms for  the $C_4$-free $2$-matching problem in bipartite graphs were shown by Hartvigsen~\cite{Hartvigsen2006}, Pap~\cite{Pap2007} and analyzed by Kir\'aly \cite{Kiraly1999}. \mwcom{add note about vertex-induced} As for the weighted version of the square-free $2$-matching problem in bipartite graphs it was proven to be NP-hard \cite{Geelen1999, Kiraly2009} and solved by  Makai~\cite{Makai2007} and Takazawa~\cite{Takazawa2009} for the case  when  the weights of edges are vertex-induced on every square of the graph. When it comes to the square-free $2$-matching problem in general graphs, 
Nam~\cite{Nam1994} constructed a complex algorithm for it for graphs, in which all squares are vertex-disjoint.
 B\'erczi and Kobayashi \cite{BercziKobayashi2012} showed that the weighted square-free $2$-matching problem 
is NP-hard for general weights even if the given graph is cubic, bipartite and planar and
gave a polynomial algorithm that finds a  maximum weight $2$-matching that contains no squares (but it can contain triangles). In \cite{BercziKobayashi2012} the square-free $2$-matching problem is used for solving the $(n-3)$-connectivity augmentation problem.
As regards subcubic graphs, there are two other results besides those mentioned above.  B\'erczi and V\'egh~\cite{BercziVegh2010} considered the problem of finding a maximum $t$-matching (a $b$-matching such that $b(v)=t$ for each vertex $v$) which does not contain any subgraph from a given set of forbidden $K_{t,t}$ and $K_{t+1}$ in an undirected graph of degree at most $t+1$. Observe that the square-free $2$-matching problem in subcubic graphs is a special case of this problem for $t=2$. \mwcom{$C_3$- and $C_4$-free also}

The  $C_4$-free $2$-matching problem  was previously investigated only in the unweighted version by Hartvigsen and Li in \cite{HartvigsenLi2011}, who devised an $O(n^{3/2})$-algorithm. 
We present combinatorial algorithms for the weighted square-free $2$-matching problem and the weighted $C_4$-free $2$-matching problem 
for the case when the weights of edges are vertex-induced on every square of the graph and the graph is subcubic. These algorithms are similar to the one for the weighted triangle-free $2$-matching problem in subcubic graphs and have the same running time.

{\bf Related work} Some generalizations of the $C_k$-free $2$-matching problem were investigated. Recently, Kobayashi~\cite{Kobayashi2020} gave a polynomial algorithm for finding a maximum weight $2$-matching that does not contain any triangle from a given set of forbidden edge-disjoint triangles. One can also consider non-simple $b$-matchings, in which every edge $e$ may occur in more than one copy. Problems connected to
non-simple $b$-matchings are usually easier than variants with simple $b$-matchings.
Efficient algorithms for triangle-free non-simple $2$-matchings (such $2$-matchings may contain $2$-cycles) were devised by Cornu\'ejols and Pulleyback \cite{CornuejolsPulleyblank1980, CornuejolsPulleyblank1980a}, Babenko, Gusakov and Razenshteyn \cite{BabenkoEtAl2010}, and Artamonov and Babenko \cite{ArtamonovBabenko2018}.
Other results for restricted non-simple $b$-matchings appeared in \cite{Pap2009,Takazawa2017,Takazawa2017a}.

%Takazawa~\cite{Takazawa2017} proposed a framework that is a generalisation of triangle-free $2$-matching, square-free $2$-matching, even factor and arborescence problems.

\section{Preliminaries}
Let $G=(V,E)$ be an undirected  graph with vertex set~$V$ and edge set~$E$. We denote the number of vertices of $G$ by $n$ and the number of edges of $G$ by $m$. We assume that all graphs are {\bf \em simple}, i.e., they contain neither loops nor parallel edges. We denote an edge connecting vertices $v$ and $u$ by $(v,u)$.  A {\bf \em cycle} of graph~$G$ is a sequence  $c=(v_0, \ldots, v_{l-1})$ for some $l \geq 3$ of pairwise distinct vertices of~$G$ such that $(v_i,v_{(i+1) \bmod l})\in E$ for every $i\in\{0,1,\ldots,l-1\}$. We refer to $l$ as the {\bf \em length} of $c$. For a given cycle $c=(v_0, \ldots, v_{l-1})$ any edge of $G$, which  connects two vertices  of $c$ and does not occur in $c$ is called a {\bf \em diagonal} (of $c$).   For a subgraph $H$ of $G$, we denote the edge set of $H$ by $E(H)$. For an edge set $F\subseteq E$ and $v\in V$, we denote by $\deg_F(v)$ the number of edges of~$F$ incident to $v$. % For an edge $e\in E$, $w(e)$ is said to be the {\bf \em weight} of $e$. For an edge set $F\subseteq E$, we denote by $w(F)$ sum of weight of edges of~$F$, i.e., $\sum_{e\in F} w(e)$. We say that $w(F)$ is the {\bf \em weight} of $F$.

An instance of each of the three problems that we consider in the paper consists of an undirected subcubic graph~$G=(V,E)$ and a weight function $w:E\to\mathbb{R}_{\geq 0}$. In the weighted triangle-free $2$-matching problem the goal is to find a maximum weight triangle-free $2$-matching of~$G$. In the weighted square-free (resp. $C_4$-free) $2$-matching problem we additionally assume that the weights on the edges are vertex-induced on each square of $G$, i.e. for any square $s=(v_1, v_2, v_3, v_4)$ there exists a function $r: \{v_1, v_2, v_3, v_4\} \rightarrow R$ such that for any edge $e=(u,v)$ connecting two vertices of $s$ it holds that $w(e)=r(u)+r(v)$. %, i.e., among all triangle-free $2$-matchings of graph~$G$ we are looking for $2$-matching $M$ which maximize $w(M)$.
The aim in the weighted square-free (resp. $C_4$-free) $2$-matching problem is to compute a maximum weight square-free (corr. $C_4$-free) $2$-matching of~$G$.

We will use the classical notion of a $b$-matching, which is a generalization of a matching. For a vector $b\in\mathbb{N}^V$, an edge set $M\subseteq E$ is said to be a $b$-matching of~$G$ if $\deg_M(v) \leq b(v)$ for every $v\in V$. Notice that a $b$-matching with $b(v)=1$ for every $v\in V$ is a classical matching. A $b$-matching of $G$ of maximum weight can be computed in polynomial time. We refer to Lov\'asz and Plummer~\cite{LovaszPlummer2009} for further background on $b$-matchings. 

We are interested in computing a $b$-matching of a graph $G$ where we are given vectors $l,u\in\mathbb{N}^V$ and a weight function~$w:E\to\mathbb{R}$. For a vertex $v\in V$, $[l(v),u(v)]$ is said to be a {\bf \em capacity interval} of $v$. An edge set~$M\subseteq E$ is said to be an \lbmatching{} if $l(v) \leq \deg_M(v) \leq u(v)$ for every $v\in V$. An \lbmatching{} $M$ is said to be a {\bf \em maximum weight} \lbmatching{} if there is no \lbmatching{}~$M'$ of~$G$ of weight greater than $w(M)$. A maximum weight \lbmatching{} can be computed efficiently.

\begin{theorem}[\cite{Gabow1983}]
There is an algorithm that, given a graph~$G=(V,E)$, a weight function~$w:E\to\mathbb{R}$ and vectors $l,u\in\mathbb{N}^V$, in time $O((\sum_{v\in V}u(v))\min\{|E(G)|\log{|V(G)|},|V(G)|^2\})$, finds a maximum weight \lbmatching{} of $G$. 
\end{theorem}

Given an \lbmatching{} $M$ and an edge $e=(u,v) \in M$, we say that $u$ is {\bf \em matched to} $v$ in $M$. \mwcom{we can replace $E(G)$ and $V(G)$ by $E$ and $G$}

\section{Outline of the Algorithm}
The general scheme of the algorithm for each variant of the restricted $2$-matching problem is the same - we give it below.

\begin{algorithm}[H]
\begin{enumerate}[labelindent=2em,labelwidth=\widthof{\ref{itm:step3}},label=\emph{Step \arabic*.},itemindent=3.7em,leftmargin=0pt]
	\renewcommand{\theenumi}{Step \arabic{enumi}}
	\item Construct an auxiliary graph~$G'=(V',E')$ of size $O(n)$ by replacing some triangles and/or squares of $G$ with gadgets containing 
	half-edges. (Both gadgets and half-edges are defined later.)
	\item Define a weight function~$w':E'\to\mathbb{R}$ and vectors $l,u\in\mathbb{N}^{V'}$ such that $u(v) \leq 2$ for every $v\in V'$.
	\item\label{itm:step3} Compute a maximum weight \lbmatching{} $M'$ of $G'$.
	\item Construct a $2$-matching~$M$ of~$G$ by replacing all half-edges of $M'$ with some edges of~$G$ in such a way that $w(M) \geq w'(M')$.
	\item Remove the remaining triangles and/or squares from $M$ by replacing some of their edges with other ones without decreasing the weight of $M$.
\end{enumerate}
\caption{Computing a maximum weight restricted $2$-matching of a subcubic graph~$G$ given a weight function $w:E\to\mathbb{R}_{\geq 0}$.}
\label{alg:main}
\end{algorithm}

\begin{claim}\label{fact:running-time}
Algorithm~\ref{alg:main} runs in time $O(n^2\log{n})$.
\end{claim}
\begin{proof}
It will be easy to implement all steps of an Algorithm~\ref{alg:main}~except \ref{itm:step3} in linear time. Hence, the running time of our algorithm is equal to the running time of an algorithm for computing a maximum weight \lbmatching{} of~$G'$, i.e., it is equal to $O((\sum_{v\in V'}u(v))\min\{|E'|\log{|V'|},|V'|^2\})$. Recall that \mbox{$|V'|+|E'|=O(n)$} and $u(v) \leq 2$ for every $v\in V'$. Hence, the running time of \ref{itm:step3} is $O(n^2\log{n})$.  \end{proof}

Let us also remark that in the unweighted versions of the problem Algorithm~\ref{alg:main} runs in $O(n^{3/2})$. \mwcom{in ,,time''?}

\tikzset{
	triangle/.pic=
	{
		\begin{scope}[font=\Huge,shift={(0,-10/6)}]
		\coordinate (-center)		at (-0.2,10/6-0.2);
		\node (a)	at (-15/6,0)	[bigVertex,label=above left:$a$]{};
		\node (b) 	at (15/6,0)		[bigVertex,label=above right:$b$]{};
		\node (c) 	at (0,21/6)		[bigVertex,label=above right:$c$]{};
		\draw[e0] (a) -- (b);
		\draw[e0] (a) -- (c);
		\draw[e0] (b) -- (c);
		\draw[e0] (a) -- +(210:1.5);
		\draw[e0] (b) -- +(-30:1.5);
		\draw[e0] (c) -- +(0,1.5);
		\coordinate (ab)	at ($(a)!0.5!(b)$);
		\coordinate (ac)	at ($(a)!0.5!(c)$);
		\coordinate (bc)	at ($(b)!0.5!(c)$); 
		\end{scope}
	},
	triangleGadget/.pic=
	{
		\begin{scope}[font=\Huge,shift={(0,-10/2)}]
		\node (a)	at (-15/2,0)	[bigVertex,label=above left:$a$]{};
		\node (b) 	at (15/2,0)		[bigVertex,label=above right:$b$]{};
		\node (c) 	at (0,21/2)		[bigVertex,label=above right:$c$]{};
		\node (vab) at (-5/2,0)		[vertex,label=below:$v^a_b$]{};
		\node (vac) at (-10/2,7/2)	[vertex,label=above left:$v^a_c$]{};
		\node (vba) at (5/2,0)		[vertex,label=below:$v^b_a$]{};
		\node (vbc) at (10/2,7/2)	[vertex,label=above right:$v^b_c$]{};
		\node (vca) at (-5/2,14/2)	[vertex,label=above left:$v^c_a$]{};
		\node (vcb) at (5/2,14/2)	[vertex,label=above right:$v^c_b$]{};
		\node (ua)	at (-6/2,4.5/2)	[vertex,label=below left:$u_a$]{};
		\node (ub)	at (6/2,4.5/2)	[vertex,label=below right:$u_b$]{};
		\node (uc)	at (0,13/2)		[vertex,label=above:$u_c$]{};
		\node (ut)	at (0,7.33/2)	[vertex,label=below:$u_t$]{};
		\draw[e0] (a) -- +(210:3/2);
		\draw[e0] (b) -- +(-30:3/2);
		\draw[e0] (c) -- +(0,3/2);
		\foreach \i in {a,b,c}
			\draw[e0] (u\i) -- (ut);
		\foreach \i in {b,c}
		{
			\draw[e0] (a) -- (va\i);
			\draw[e0] (va\i) -- (ua);
		}
		\foreach \i in {a,c}
		{
			\draw[e0] (b) -- (vb\i);
			\draw[e0] (vb\i) -- (ub);
		}
		\foreach \i in {a,b}
		{
			\draw[e0] (c) -- (vc\i);
			\draw[e0] (vc\i) -- (uc);
		}
		\draw[e0] (vab) -- (vba);
		\draw[e0] (vac) -- (vca);
		\draw[e0] (vbc) -- (vcb);
		\end{scope}
	},
	triangleGadgetOpt/.pic=
	{
		\pic {triangleGadget};
		\draw[e1] (vab) -- (ua);
		\draw[e1] (vba) -- (ub);
		\draw[e1] (ut) -- (uc);
	},
	triangleGadgetNoeliminator/.pic=
	{
		\pic {triangleGadget};
		\draw[e1] (c) -- (vca);
		\draw[e1] (c) -- (vcb);
		\draw[e1] (ut) -- (uc);
	},
	pics/triangleHalfEdges/.style n args={6}{code={
		\pic {triangle};
		\draw[#1] (a) -- (ab);
		\draw[#2] (ab) -- (b);
		\draw[#3] (a) -- (ac);
		\draw[#4] (ac) -- (c);
		\draw[#5] (b) -- (bc);
		\draw[#6] (bc) -- (c);
	}},
	pics/triangleToOpt/.style n args={9}{code={
		\pic (arrow) [scale=0.5]{transformsTo};
		\pic[left=5.5cm of arrow-leftEnd] {triangleHalfEdges={#1}{#2}{#3}{#4}{#5}{#6}};
		\pic[right=5.5cm of arrow-rightEnd] {triangleHalfEdges={#7}{#7}{#8}{#8}{#9}{#9}};
	}}
}

\section{Triangle-free $2$-matchings in subcubic graphs} \label{sec:simple}

In this section we solve a maximum weight triangle-free $2$-matching problem in subcubic graphs. We assume that each connected component of $G$ is different from $K_4$, i.e., different from a $4$-vertex clique. \mwcom{we can get rid of $K_4$ assumption}

One can observe that, since $G$ is subcubic, any edge $e$ of $G$ belongs to at most two different triangles. \mwcom{we do not need this observation} Also, any triangle of $G$
shares an edge with at most one other triangle or, in other words, any triangle of $G$ is not edge-disjoint with at most one other triangle.

\begin{definition}\label{def:problematic-triangle}
A triangle $t$, which has a common edge with some other triangle $t'$ such that $w(t) \leq w(t')$ is said to be {\bf \em unproblematic}. Otherwise, $t$ is said to be {\bf \em problematic}.
\end{definition}

Unproblematic triangles can be easily got rid of from any $2$-matching $M$ of $G$ by replacing some of its edges with other ones as explained in more detail in the proof of Theorem \ref{thm:single-triangles}.

Observe that any problematic triangle of $G$ is vertex-disjoint with any other problematic triangle of $G$.

We begin with the following simple fact. \mwcom{add proof?}
\begin{claim} \label{fact:incircle}
Let $t=(a,b,c)$ be a triangle of $G$, whose edges have weights $w(a,b), w(b,c), w(c,a)$, respectively. Then, there exist real numbers $r_a,r_b,r_c$ such that %$w(a,b)=x+y, w(b,c)=y+z, w(c,a)=x+z$. 
$w(a,b)=r_a+r_b$, $w(b,c)=r_b+r_c$ and $w(c,a)=r_c+r_a$.
\end{claim}

If the weights of edges of $t$ satisfy the triangle inequality, then Claim~\ref{fact:incircle} has a geometric interpretation connected to how an  incircle partitions the edges of a triangle - see Figure \ref{incircle}. \mwcom{"Also, the relationship to geometry in the case when the triangle inequality holds is unclear and needs more detail"}

\begin{figure}
\centering{\includegraphics[scale=0.85]{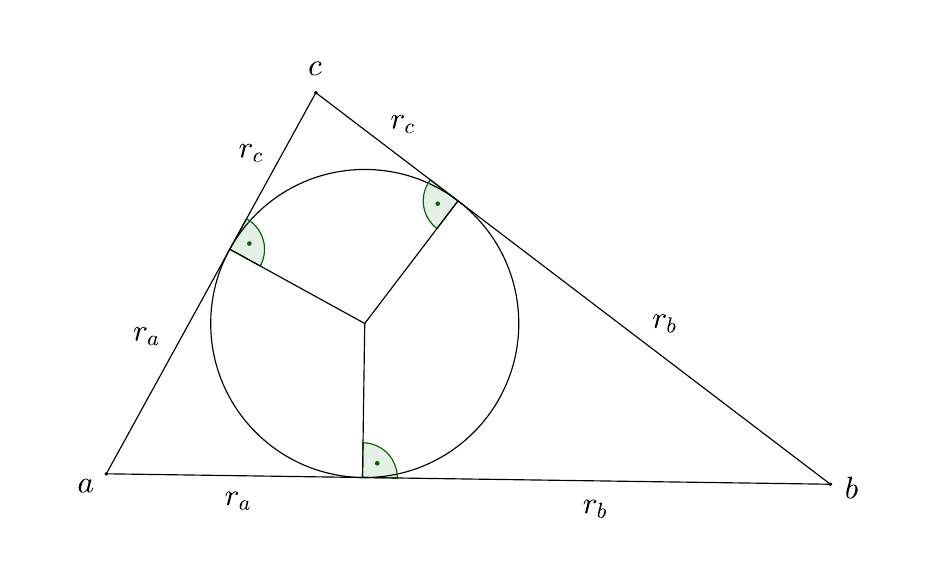}}
\caption{\label{incircle} Partition of the edges of a triangle by its incircle.}
\end{figure}

If $G$ contains at least one problematic triangle, we build a graph $G'=(V', E')$ together with a weight function \mbox{$w':E'\to\mathbb{R}$}, in which each problematic triangle $t$ is replaced with a subgraph, called a {\bf \em gadget for $t$}. The precise construction of $G'$ is the following. We start off with $G$.

Let $t=(a,b,c)$ be any problematic triangle of $G$. For each edge $(p,q)$ of $t$ we add two new vertices $v^p_q$ and $v^q_p$, called {\bf \em subdivision vertices} (of $t$), and we replace $(p,q)$
with three new edges: $(p, v^p_q), (v^p_q, v^q_p),  (v^q_p, q)$. Each of the edges $(p, v^p_q),  (v^q_p, q)$ is called a {\bf \em half-edge}
(of $(p,q)$ and also of $t$). The edge $(v^p_q, v^q_p)$ is called an {\bf \em eliminator} (of $(p,q)$). The half-edges of edges of $t$ get weights equal to values of $r_a, r_b, r_c$ from Claim~\ref{fact:incircle}, i.e., $w'(a, v^a_b)=w'(a, v^a_c)=r_a$, $w'(b, v^b_a)=w'(b,v^b_c)=r_b$ and $w'(c,v^c_a)=w'(c,v^c_b)=r_c$. The weight of each eliminator is $0$. Additionally,
we introduce four new vertices $u_a$, $u_b$, $u_c$, $u_t$, called {\bf \em global vertices}. For every $d\in\{a,b,c\}$ we connect $u_d$ with $u_t$ and with every subdivision vertex connected to $d$. Every edge incident to a global vertex  has weight  $0$.

We define vectors $l,u\in\mathbb{N}^{V'}$ as follows. We set a capacity interval of every vertex of the original graph~$G$ to $[0,2]$ and we set a capacity interval of every other vertex of $G'$ to $[1,1]$, i.e., every vertex of $V'\setminus V$ is matched to exactly one vertex of $G'$ in any \lbmatching{} of $G'$.

The main ideas behind  the gadget for a problematic triangle $t=(a,b,c)$ are the following. An \lbmatching{} $M'$ of $G'$ is to represent roughly a triangle-free $2$-matching $M$ of $G$. If $M'$ contains both half-edges of some edge $e$, then $e$ is included in $M$. If $M'$ contains an eliminator of $e$,
then $e$ does not belong to $M$ (is excluded from $M$).  We want to ensure that at least one edge of $t$ does not belong to $M$. This is done
by requiring that two of the global vertices $u_a, u_b, u_c$ are matched to subdivision vertices. In this way two half-edges of $t$ are guaranteed not to belong to $M'$ and hence to $M$.

\begin{figure}[htpb]
\centering
\begin{tikzpicture}[scale=0.35,transform shape]
	\pic (arrow) [scale=0.5]{transformsTo};
	\pic (t) [left=7cm of arrow-leftEnd] {triangle};
	\begin{scope}[font=\Huge]
	\path (t-center) node{$t$};
	\end{scope}
	
	\pic (g) [right=10cm of arrow-rightEnd] {triangleGadget};
	\begin{scope}[font=\Huge,color=red]
	\path (ga) -- (gvab) node[midway,below] {$r_a$};
	\path (ga) -- (gvac) node[midway,above left] {$r_a$};
	\path (gb) -- (gvba) node[midway,below] {$r_b$};
	\path (gb) -- (gvbc) node[midway,above right] {$r_b$};
	\path (gc) -- (gvca) node[midway,above left] {$r_c$};
	\path (gc) -- (gvcb) node[midway,above right] {$r_c$};
	\end{scope}
\end{tikzpicture}
\caption{A gadget for a problematic triangle $t=(a,b,c)$.}
\end{figure}
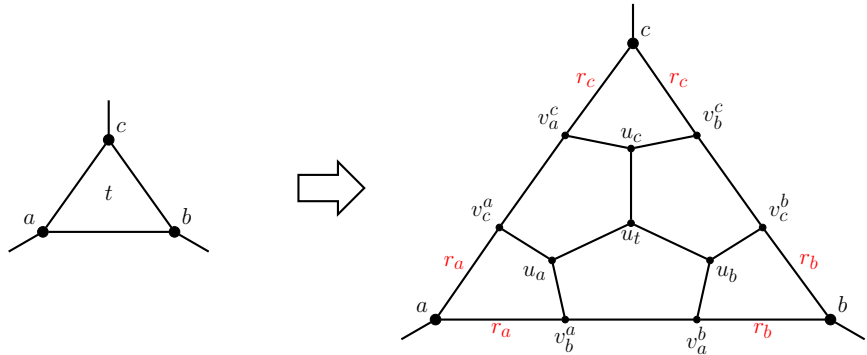

In the theorem below we show the correspondence between triangle-free $2$-matchings of $G$ and {\lbmatching{}}s of $G'$.

\begin{theorem}\label{lem:opt-to-matching-simple}
Let $M$ be any triangle-free $2$-matching of $G$. Then we can find an \lbmatching{} $M'$ of $G'$ such that $w'(M') = w(M)$. 
\end{theorem}
\begin{proof}
We initialize $M'$ as the empty set. We add every edge of $M$ that does not belong to any problematic triangle of $G$ to $M'$. 
Consider any problematic triangle $t=(a,b,c)$ of $G$.  Since $M$ is triangle-free, there exists an edge of $t$ that does not belong to $M$. If more than one edge of $t$ does not belong to $M$, we choose  one of them. Suppose that we chose $(a,b)\notin M$. Then we add  edges $(v^a_b,u_a)$, $(v^b_a,u_b)$ and $(u_t,u_c)$ to $M'$. For every other edge $e$ of $t$ we proceed as follows. If $e \in M$, we add both  half-edges of $e$ to $M'$, otherwise we add the eliminator of $e$ to $M'$. %See Figure~\ref{fig:opt-configurations-simple} for examples.
Since the weight of any edge of $t$ in $G$ is equal to the sum of the weights of its half-edges in $G'$, we get that $w'(M') = w(M)$.
\end{proof}

\begin{comment}
\tikzset{
	e0/.style={line width=0.4pt},
	e1/.style={line width=0.5mm}}
\begin{figure}[htpb]
\centering
\begin{subfigure}{0.32\textwidth}
	\centering
	\begin{tikzpicture}[scale=0.25,transform shape]
	\pic {triangleGadgetOpt};
	\draw[e1] (vac) -- (vca);
	\draw[e1] (vbc) -- (vcb);
	\end{tikzpicture}
	\caption{$E(t)\cap M=\emptyset$}
\end{subfigure}
\begin{subfigure}{0.32\textwidth}
	\centering
	\begin{tikzpicture}[scale=0.25,transform shape]
	\pic {triangleGadgetOpt};
	\draw[e1] (a) -- (vac);
	\draw[e1] (c) -- (vca);
	\draw[e1] (vbc) -- (vcb);
	\end{tikzpicture}
	\caption{$E(t)\cap M=\{\{a,c\}\}$}
\end{subfigure}
\begin{subfigure}{0.32\textwidth}
	\centering
	\begin{tikzpicture}[scale=0.25,transform shape]
	\pic {triangleGadgetOpt};
	\draw[e1] (a) -- (vac);
	\draw[e1] (c) -- (vca);
	\draw[e1] (c) -- (vcb);
	\draw[e1] (b) -- (vbc);
	\end{tikzpicture}
	\caption{$E(t)\cap M=\{\{a,c\},\{b,c\}\}$}
\end{subfigure}
\caption{\label{fig:opt-configurations-simple} The construction of an \lbmatching{} of $G'$ corresponding to a triangle-free $2$-matching $M$ of $G$. Thicker lines represent edges of a gadget corresponding to a triangle $t$ of~$G$ that are added to \lbmatching{}. Subpoints corresponds to different configurations of edges of $M$ which belongs to $t$. We proceed analogously in other cases.}
\end{figure}
\tikzset{
	e0/.style={line width=0.8pt},
	e1/.style={line width=0.8mm}}

\end{comment}

\begin{theorem}\label{thm:single-triangles}
Let $M'$ be any \lbmatching{} of $G'$. Then we can find a triangle-free $2$-matching~$M$ of $G$ such that $w(M) \geq w'(M')$.
\end{theorem}
\begin{proof}
%We will construct a triangle-free $2$-matching $M$ of $G$ such that $w(M)\geq w'(M')$. By Theorem~\ref{lem:opt-to-matching-simple}, it  implies that $w(M)=w'(M')$.

We initialize $M$ as the empty set. We add every edge of $M'$ that belongs to $G$ to $M$. For every problematic triangle of $G$ we will add some of its edges to $M$.

Consider any problematic triangle $t=(a,b,c)$ of $G$. Notice that exactly two of the vertices $u_a, u_b, u_c$ are matched to subdivision vertices, because  $u_t$ is matched to one of $u_a, u_b, u_c$.   This corresponds to excluding two half-edges of $t$ from  $M$. Since every subdivision vertex is required to be matched to exactly one vertex in $G'$, 
we get that an even number and at most four subdivision vertices of $t$ are matched to the vertices $a,b,c$. This indicates,  which half-edges of $t$ are going to be included in $M$.  Observe that  the two subdivision vertices that are matched in $M'$ to vertices $u_a, u_b, u_c$ are  adjacent to two different   vertices of $t$. Thus, we have: \mwcom{add proof?}
\begin{claim} \label{4polkr}
If  $M'$ contains exactly four half-edges of $t$, then  the two half-edges of $t$ that do not belong to $M'$ are not adjacent to the same vertex of $t$. \mwcom{maybe define "adjacent" for half-edges? Shouldn't it be "incident" instead of "adjacent"?} 
\end{claim}
Every other subset of half-edges of $t$ containing an even number of at most four half-edges of $t$ can occur in $M'$.

In each of these cases, we proceed as follows (see Figure~\ref{fig:matching-to-opt-simple}):
\begin{enumerate}
\item Exactly zero subdivision vertices of $t$ are matched to $a,b,c$. We do not include any edge of $t$ in $M$.
\item Exactly two subdivision vertices of $t$ are matched to $a,b,c$.
\begin{enumerate}
\item The two subdivision vertices of $t$ are matched to two different vertices $u,v$ of $t$. Then we include the edge $(u,v)$ in $M$.
\item The two subdivision vertices of $t$ are matched to the same  vertex $u$ of  $t$. Then we include in $M$ two  edges of $t$ incident to $u$. (This is the only case where $w(M)$ may be greater than $w'(M')$ when it comes to half-edges of $t$. Notice that for any two vertices $u,v$ of $t$ we have that $r_u+r_v \geq 0$).
\end{enumerate}
\item Exactly four subdivision vertices of $t$ are matched to $a,b,c$. Then, by Claim~\ref{4polkr}, two of these vertices are matched to the same vertex $u$ of $t$
and the other two are matched to the remaining two vertices of $t$. In this case we include in $M$ two  edges of $t$ incident to $u$.

\end{enumerate}

%Note that Figure~\ref{fig:matching-to-opt-simple} does not show cases, in which for each edge $e$ of a problematic triangle $t$, $M'$ contains either both half-edges of $e$ or the eliminator of $e$. 

Since half-edges incident to the same vertex have the same weight, we get that $w(M)\geq w'(M')$.

The resulting $2$-matching $M$ can contain some unproblematic triangles. We remove them one by one. Let $t=(a,b,c)$ be any such triangle. From Definition~\ref{def:problematic-triangle} there exists  another triangle $t'=(a,b,d)$, which shares an edge with $t$ and such that $w(t') \geq w(t)$. Hence, either $w(a,d)\geq w(a,c)$ or $w(b,d)\geq w(b,c)$. Assume that $w(a,d)\geq w(a,c)$. We replace the edge $(a,c)$ with the edge $(a,d)$ \mwcom{should it be a separate observation at the beginning?} without decreasing the weight of $M$.
\end{proof}

\tikzset{
	e0/.style={line width=0.4pt},
	e1/.style={line width=0.5mm}}
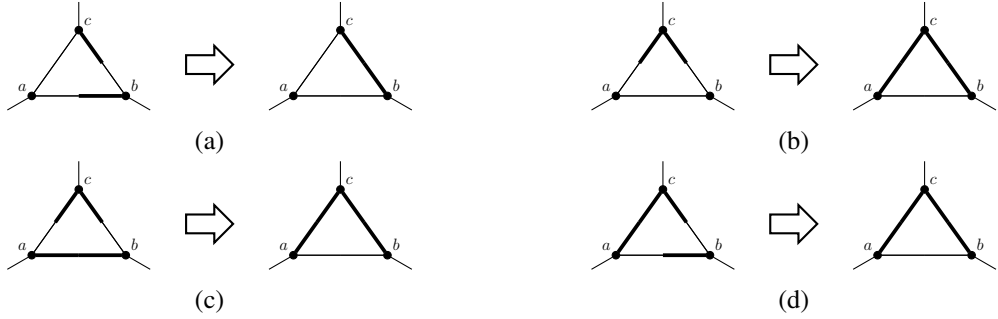
\begin{figure}[htpb]
\centering
\begin{subfigure}{0.49\textwidth}
	\centering
	\begin{tikzpicture}[scale=0.25,transform shape]
	\pic {triangleToOpt={e0}{e1}{e0}{e0}{e0}{e1}{e0}{e0}{e1}};
	\end{tikzpicture}
	\caption{}
\end{subfigure}
\begin{subfigure}{0.49\textwidth}
	\centering
	\begin{tikzpicture}[scale=0.25,transform shape]
	\pic {triangleToOpt={e0}{e0}{e0}{e1}{e0}{e1}{e0}{e1}{e1}};
	\end{tikzpicture}
	\caption{\label{fig:matching-to-opt-simple-f}}
\end{subfigure}
\begin{subfigure}{0.49\textwidth}
	\centering
	\begin{tikzpicture}[scale=0.25,transform shape]
	\pic {triangleToOpt={e1}{e1}{e0}{e1}{e0}{e1}{e0}{e1}{e1}};
	\end{tikzpicture}
	\caption{}
\end{subfigure}
\begin{subfigure}{0.49\textwidth}
	\centering
	\begin{tikzpicture}[scale=0.25,transform shape]
	\pic {triangleToOpt={e0}{e1}{e1}{e1}{e0}{e1}{e0}{e1}{e1}};
	\end{tikzpicture}
	\caption{}
\end{subfigure}
\caption{\label{fig:matching-to-opt-simple} The construction of a maximum weight triangle-free $2$-matching of $G$ from a maximum weight \lbmatching{} of $G'$. }
\end{figure}
\tikzset{
	e0/.style={line width=0.8pt},
	e1/.style={line width=0.8mm}}

\tikzset
{
	vertex/.default=1.5mm,
	bigVertex/.style={vertex=2.4mm},
	singleSquare/.pic=
	{
		\begin{scope}[font=\Large]
		\node (a)	at (-1,1)		[vertex,label=above:$a$]{};
		\node (b)	at (1,1)		[vertex,label=above:$b$]{};
		\node (c)	at (1,-1)		[vertex,label=below:$c^{\phantom{0}}$]{};
		\node (d)	at (-1,-1)		[vertex,label=below:$d$]{};
		\node (aa)	at (-1.5,1.5)	[]{};
		\node (bb)	at (1.5,1.5)	[]{};
		\node (cc)	at (1.5,-1.5)	[]{};
		\node (dd)	at (-1.5,-1.5)	[]{};
		\path (0.15,0) node{$s$};
		\end{scope}
		
		\foreach \x in {a,b,c,d}
			\draw[e0] (\x) -- (\x\x);
		\draw[e0] (a) -- (b) -- (c) -- (d) -- (a);
	},
	singleSquareGadget/.pic=
	{
		\begin{scope}[font=\Large]
		\node (a)	at (-3,3)		[bigVertex,label=above:$a$]{};
		\node (b)	at (3,3)		[bigVertex,label=above:$b$]{};
		\node (c)	at (3,-3)		[bigVertex,label=below:$c^{\phantom{0}}$]{};
		\node (d)	at (-3,-3)		[bigVertex,label=below:$d$]{};
		\node (aa)	at (-4,4)		[]{};
		\node (bb)	at (4,4)		[]{};
		\node (cc)	at (4,-4)		[]{};
		\node (dd)	at (-4,-4)		[]{};
		\node (vab) at (-1,3)		[vertex,label=above:$v^a_b$]{};
		\node (vba) at (1,3)		[vertex,label=above:$v^b_a$]{};
		\node (vbc) at (3.05,1)		[vertex,label=right:$v^b_c$]{};
		\node (vcb) at (3.05,-1)	[vertex,label=right:$v^c_b$]{};
		\node (vcd) at (1,-3)		[vertex,label=below:$v^c_d$]{};
		\node (vdc) at (-1,-3)		[vertex,label=below:$v^d_c$]{};
		\node (vda) at (-2.95,-1)	[vertex,label=left:$v^d_a$]{};
		\node (vad) at (-2.95,1)	[vertex,label=left:$v^a_d$]{};
		\node (us1)	at (-1,1)		[vertex,label=above left:$u_s^1$]{};
		\node (us2)	at (1,1)		[vertex,label=above right:$u_s^2$]{};
		\end{scope}
		
		\foreach \x in {a,b,c,d}
			\draw[e0] (\x) -- (\x\x);
		\draw[e0] (a) -- (vab) node[midway,red,above] {$r_a$};
		\draw[e0] (a) -- (vad) node[midway,red,left] {$r_a$};
		\draw[e0] (b) -- (vbc) node[midway,red,right] {$r_b$};
		\draw[e0] (b) -- (vba) node[midway,red,above] {$r_b$};
		\draw[e0] (c) -- (vcd) node[midway,red,below] {$r_c$};
		\draw[e0] (c) -- (vcb) node[midway,red,right] {$r_c$};
		\draw[e0] (d) -- (vda) node[midway,red,left] {$r_d$};
		\draw[e0] (d) -- (vdc) node[midway,red,below] {$r_d$};
		\draw[e0] (vab) -- (vba);
		\draw[e0] (vbc) -- (vcb);
		\draw[e0] (vcd) -- (vdc);
		\draw[e0] (vda) -- (vad);
		\foreach \x in {b,d}
			\draw[e0] (va\x) -- (us1) -- (vc\x);
		\foreach \x in {a,c}
			\draw[e0] (vb\x) -- (us2) -- (vd\x);
	},
	pics/tripleSquare/.style n args={6}{code={
		\node (a) at (-1,0) [vertex,label=below:$a$]{};
		\node (b) at (0,0) [vertex,label={[label distance=-3]below right:$b$}]{};
		\node (c) at (1,0) [vertex,label=below:$c$]{};
		\node (d) at (0,1) [vertex,label=above:$d$]{};
		\node (e) at (0,-1) [vertex,label=below:$e$]{};
		\node (aa) at (-1.5,0) []{};
		\node (bb) at (0.5,0) []{};
		\node (cc) at (1.5,0) []{};
		
		\foreach \x in {a,b,c}
			\draw[e0] (\x) -- (\x\x);
		\draw[#1] (a) -- (d);
		\draw[#2] (a) -- (e);
		\draw[#3] (b) -- (d);
		\draw[#4] (b) -- (e);
		\draw[#5] (c) -- (d);
		\draw[#6] (c) -- (e);
	}},
	tripleSquareGadget/.pic=
	{
		\begin{scope}[font=\large,vertex/.default=1mm,bigVertex/.style={vertex=1.8mm}]
		\node (us)	at (0.025,0)	[vertex,label=below:$u_S$,label={[font=\footnotesize,gray]above right:$[0,1]$}]{};
		\node (a)	at (210:2)		[bigVertex,label=below:$a$,label={[font=\footnotesize,gray,label distance=5]right:$[0,1]$}]{};
		\node (b)	at (330:2)		[bigVertex,label=below:$b$,label={[font=\footnotesize,gray,label distance=5]left=1:$[0,1]$}]{};
		\node (c)	at (0,2)		[bigVertex,label=-5:$c$,label={[font=\footnotesize,gray]above right:$[0,1]$}]{};
		\node (aa)	at (210:2.75)	[]{};
		\node (bb)	at (330:2.75)	[]{};
		\node (cc)	at (-0.04,2.75)	[]{};
		
		\node (us1) at (4.5,1) [vertex,label=left:$u_S^1$,label={[font=\footnotesize,gray]above right:$[1,1]$}]{};
		\node (us2) at (4.5,0) [vertex,label=left:$u_S^2$,label={[font=\footnotesize,gray]below right:$[1,1]$}]{};
		\end{scope}
		
		\begin{scope}[font=\scriptsize]
		\foreach \x in {a,b,c}
			\draw[e0] (\x) -- (\x\x);
		\draw[e0] (a) -- (us) node[red,midway,above left] {$r_a+r_e$};
		\draw[e0] (b) -- (us) node[red,midway,above right] {$r_b+r_e$};
		\draw[e0] (c) -- (us) node[red,midway,left] {$r_c+r_e$};
		
		\draw[e0] (us1) -- (us2) node[red,midway,right] {$R$};
		\end{scope}
	}
}

\section{Square-free $2$-matchings in subcubic graphs} \label{sec:squares}
In this section we solve a maximum weight square-free $2$-matching problem in subcubic graphs. Recall that this problem is NP-hard for general weights, therefore we assume that weights are vertex-induced on every square, i.e., for any square $s=(a,b,c,d)$ of $G$ there exist real numbers $r_a, r_b, r_c, r_d$, called {\bf \em potentials} of $s$ such that for any edge $e=(u,v)$ connecting two vertices of $s$ \mwcom{stress that it holds for native edges only, not diagonals} it holds that $w(e)=r_u+r_v$. (Note that if a given edge $e=(u,v)$ belongs to two different squares $s$ and $s'$, then potentials of $s$ and $s'$ on $u$ and $v$ may be different.)

We also assume  that each connected component of $G$ is different from $K_4$.

For a square $s=(v_0,v_1,v_2,v_3)$ of $G$, edges $(v_0,v_1)$, $(v_1,v_2)$, $(v_2,v_3)$ and $(v_3,v_0)$ are said to be {\em \bf native} edges of $s$. \mwcom{maybe move definition of native edges to preliminaries?} %an edge $(v_i,v_{(i+1) \bmod 4})$ is said to be a {\em \bf native} edge of $s$ for every $i\in\{0,1,2,3\}$.

One can observe that, since $G$ is subcubic, any two different squares of $G$ are vertex-disjoint or have either one or two edges in common.

\begin{definition}\label{def:problematic-square}
A square $s$ of $G$ is said to be {\em \bf unproblematic} if there exists another square $s'$ such that (i) $s$ shares exactly one edge with $s'$ or (ii) $s$ shares two edges with $s'$ and $w(s) \leq w(s')$. Otherwise, $s$ is said to be {\em \bf problematic}.
\end{definition}

Observe that any problematic square of $G$ is vertex-disjoint with any other problematic square of $G$.

The following simple observation shows that squares which have exactly one common edge with another square do not pose any problem for computing a maximum weight square-free $2$-matching of $G$.

\begin{claim}\label{obs:one-common-edge-squares}
Consider any two squares $s=(a,b,c,d)$ and $s'=(c,d,e,f)$ of $G$ which share exactly one  edge. Let $M_1$ be a $2$-matching of $G$ that contains $s$. Then there exists a $2$-matching $M_2$ of $G$, which does not contain $s$ or any square not already contained in $M_1$ and such that  $w(M_2) \geq w(M_1)$.
\end{claim}
\begin{proof}
We set $M_2=M_1\setminus\{(c,d),(e,f)\}\cup\{(c,f),(d,e)\}.$ Note that we can assume that $M_1$ contains the edge $(e,f)$, \mwcom{note that adding $(e,f)$ to $M_1$ can introduce some square} because $G$ is subcubic and $M_1$ contains neither $(c,f)$ nor $(d,e)$.
It is straightforward to check that $M_2$ is a $2$-matching of $G$ that does not contain $s$. Furthermore, the given construction does not introduce any additional squares into $M_2$. Observe that $w(M_2) \geq w(M_1)$, since $w$ is vertex-induced on $s'$. 
\end{proof}

We show the construction of a gadget for a problematic square $s=(a,b,c,d)$.  We use the notation introduced in Section \ref{sec:simple}. For every native edge $(p,q)$ of $s$, we introduce two subdivision vertices $v^p_q,v^q_p$ and replace  $(p,q)$  with two half-edges $(p,v^p_q)$ and $(v^q_p,q)$ and an eliminator $(v^p_q,v^q_p)$. (We do not replace any diagonal of $s$.) Additionally, we introduce two new global vertices $u_s^1$ and $u_s^2$. We connect $u_s^1$ with all subdivision vertices adjacent to either $a$ or $c$. Symmetrically, we connect $u_s^2$ with all subdivision vertices adjacent to either $b$ or $d$.
 
The half-edges incident to $a$, $b$, $c$ and $d$ get weight $r_a$, $r_b$, $r_c$ and $r_d$, respectively, where $r_a, \ldots, r_d$  are potentials of $s$.  All other edges of the gadget get weight $0$. We set a capacity interval of every vertex of $s$ to $[0,2]$ and we set a capacity interval of every other vertex of the gadget to $[1,1]$.

\begin{figure}[htpb]
\centering
\begin{tikzpicture}[scale=0.75,transform shape]
	\pic (arrow) [scale=0.3]{transformsTo};
	\pic (ss) [left=4cm of arrow-leftEnd] {singleSquare};
	\pic (g) [right=6cm of arrow-rightEnd] {singleSquareGadget};
\end{tikzpicture}
\caption{A gadget for a problematic square $s=(a,b,c,d)$.}
\end{figure}
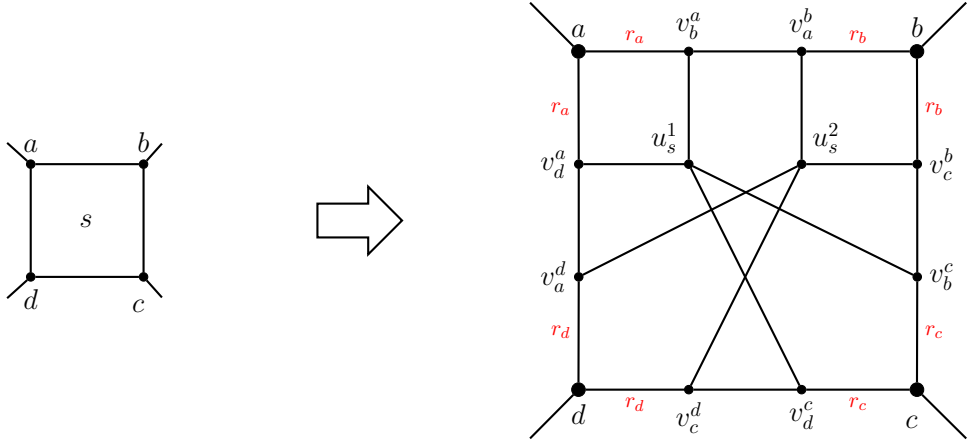

\begin{theorem}\label{lem:square-opt-to-matching}
Let $M$ be any square-free $2$-matching of $G$. Then we can find an \lbmatching{} $M'$ of $G'$ such that $w'(M') = w(M)$.
\end{theorem}
\begin{proof}
We initialize $M'$ as the empty set. We add every edge of $M$ that does not belong to any problematic square of $G$ to $M'$.

Consider any problematic square $s=(a,b,c,d)$ of $G$. Assume that $(a,b)$ does not belong to $M$. We add edges $(v^a_b,u_s^1)$ and $(v^b_a,u_s^2)$ to $M'$. For every other native edge $e$ of $s$ we proceed as follows. If $e\in M$, we add both half-edges of $e$ to $M'$, otherwise we add the eliminator of $e$ to $M'$.

%Consider a triple square $S=(a,b,c,d,e)$. We add an edge $(u_S^1;u_S^2)$ to $M'$. Additionally, if case~\ref{itm:triple-square-opt-2} of Observation~\ref{obs:triple-square-opt} with regard to $S$ and $M$ holds, we add $(w;u_S)$ to $M'$.
\end{proof}

\begin{theorem}\label{thm:square-matching-to-opt}
Let $M'$ be any \lbmatching{} of $G'$. Then we can find a square-free $2$-matching $M$ of $G$ such that $w(M) \geq w'(M')$.
\end{theorem}
\begin{proof}
We initialize $M$ as the empty set. We add every edge of $M'$ that belongs to $G$ to $M$. For every problematic square of $G$ we will add some of its edges to $M$. Next we will replace some edges of $M$ with other ones to remove unproblematic squares.

Consider any problematic square $s=(a,b,c,d)$ of $G$. Notice that there exists a native edge $(p, q)$ of $s$ such that $u_s^1$ and $u_s^2$ are matched in $M'$ to two subdivision vertices, one of which is adjacent to $p$ and the other to $q$. W.l.o.g. assume that $(p,q)=(a,b)$. We consider the following cases:
\begin{enumerate}
\item\label{itm:single-square-case-1} $u_s^1$ and $u_s^2$ are matched in $M'$ to $v^a_b$ and $v^b_a$, respectively. We add every native edge of $s$ whose both half-edges belong to $M'$ to $M$. Notice that for every other native edge $e$ of $s$, the eliminator of $e$ belongs to $M'$.
\item\label{itm:single-square-case-2} Either $u_s^1$ is matched to $v^a_b$ or $u_s^2$ is matched to $v^b_a$ in $M'$, but not both of them. Assume that $u_s^1$ is matched to $v^a_b$. Therefore, edges $(u_s^2,v^b_c)$ and $(b,v^b_a)$ belong to $M'$. We replace these two edges with $(u_s^2,v^b_a)$ and $(b,v^b_c)$ without changing the weight of $M'$. Then we proceed as in  case~\ref{itm:single-square-case-1}.
\item $u_s^1$ and $u_s^2$ are matched to $v^a_d$ and $v^b_c$, respectively, in $M'$. If $(v^a_b,v^b_a)$ does not belong to $M'$, we connect $u_s^1$ and $u_s^2$ with $v^a_b$ and $v^b_a$, respectively, similarly as in  case~\ref{itm:single-square-case-2}, and we proceed as in  case~\ref{itm:single-square-case-1}.  Assume now that $(v^a_b,v^b_a)$ belongs to $M'$. Notice that $(d,v^d_a)$ and $(c,v^c_b)$ belong to $M'$. We add $(a,d)$ and $(b,c)$ to $M$. Additionally, if both half-edges of $(c,d)$ belong to $M'$, we add $(c,d)$ to $M$.
\end{enumerate}

%Consider any triple square $S=(a,b,c,d,e)$ of $G$. Assume that potential of $d$ is greater than or equal to potential of $e$. If $u_S$ is matched to some $w\in\{a,b,c\}$ in $M'$, we connect $w$ with $d$ and $e$ in $M$, we connect one vertex of $\{a,b,c\}\setminus \{w\}$ with $d$ and another vertex with $e$. Otherwise, we add $(a,d)$, $(b,d)$ and $(c,e)$ to $M$.

The resulting $2$-matching $M$ can contain some unproblematic squares. We remove squares, which share exactly one edge with  another square from $M$ \mwcom{stress that "from M" applies to "squares" and not to "another square"} one by one using Claim~\ref{obs:one-common-edge-squares}. We remove the rest of unproblematic squares in a similar way as we got rid of   unproblematic triangles in the proof of Theorem~\ref{thm:single-triangles}. Each such removal does not introduce any squares into $M$, therefore $M$ is a square-free $2$-matching in the end.
\end{proof}

\tikzset{
	doubleTriangle/.pic=
	{
		\begin{scope}[font=\Large]
		\node (a)	at (0,1)		[vertex,label=above left:$a$]{};
		\node (b)	at (0,-1) 		[vertex,label=below left:$b$]{};
		\node (c)	at (-1,0)		[vertex,label=left:$c$]{};
		\node (d)	at (1,0) 		[vertex,label=right:$d$]{};
		\node (aa)	at (0.05,1.75)	[]{};
		\node (bb)	at (0.05,-1.75)	[]{};
		\path (0.25,0.4) node{$t_1$};
		\path (0.25,-0.4) node{$t_2$};
		\end{scope}
		
		\foreach \x in {a,b}
			\draw[e0] (\x) -- (\x\x);
		\draw[e0] (c) -- (d);
		\draw[e0] (c) -- (a);
		\draw[e0] (c) -- (b);
		\draw[e0] (d) -- (a);
		\draw[e0] (d) -- (b);
		
	},
	doubleTriangleGadget/.pic=
	{
		\begin{scope}[font=\Large]
		\coordinate (-upCenter)		at (0,0.8);
		\coordinate (-downCenter)	at (0,-0.8);
		\coordinate (-center)		at (0,0);
		\coordinate (-c)			at (-1,0);
		\coordinate (-d)			at (1,0);
		\node (a)	at (0,2)		[vertex,label={right:$a$},label={left:\large $[0,1]$}]{};
		\node (b)	at (0,-2) 		[vertex,label={right:$b$},label={left:\large $[0,1]$}]{};
		\node (aa)	at (-0.05,2.75)	[]{};
		\node (bb)	at (-0.05,-2.75)[]{};
		\node (ut)	at (0,0)		[vertex,label={right:$u_T$},label={left:\large $[0,1]$}]{};
		\node (vt1) at (4,1)		[vertex,label={right:$v_T^1$},label={left:\large $1$}]{};
		\node (vt2) at (4,-1)		[vertex,label={right:$v_T^2$},label={left:\large $1$}]{};
		\end{scope}
		
		\foreach \x in {a,b}
			\draw[e0] (\x) -- (\x\x);
		\begin{scope}[font=\small]
		\draw[e0] (vt1) -- (vt2)	node[midway,right,red] {$M_1^1(T)$};
		\draw[e0] (ut) -- (a)		node[midway,right,red] {$M^2_1(T)-M^1_1(T)$};
		\draw[e0] (ut) -- (b)		node[midway,right,red] {$M^1_2(T)-M^1_1(T)$};
		\end{scope}
	}
}

\section{$C_4$-free $2$-matchings in subcubic graphs}
In this section we solve a maximum weight $C_4$-free $2$-matching problem in subcubic graphs. We assume that weights are vertex-induced on every square. We also assume  that each connected component of $G$ is different from $K_4$.

We say that a cycle $C$ of $G$ is {\bf \em short} if it is either a triangle or a square. We say that a short cycle $C$ of $G$ is {\bf \em unproblematic} if it shares exactly one edge with some square of $G$ or if it fits  Definition~\ref{def:problematic-triangle} or Definition~\ref{def:problematic-square}. A short cycle, which is not unproblematic is said to be {\bf \em problematic}.

We have the analogue of Claim \ref{obs:one-common-edge-squares}, which justifies considering triangles sharing one edge with a square unproblematic:
\begin{claim}\label{obs:one-common-edge-squares1}
Consider  two short cycles: a triangle  $t=(a,c,d)$ and a square $s'=(c,d,e,f)$ of $G$ which share exactly one  edge. Let $M_1$ be a $2$-matching of $G$ that contains $t$. Then there exists a $2$-matching $M_2$ of $G$, which does not contain $t$ or any short cycle not already contained in $M_1$ and such that  $w(M_2) \geq w(M_1)$.
\end{claim}

Observe that any two different short problematic cycles that are not vertex-disjoint must form a pair consisting of a square $s=(a,c,b,d)$ and a triangle $t_1=(a,c,d)$ with exactly two common edges. We call a subgraph induced on vertices of such $s$ and $t_1$ a {\bf \em double triangle} $T=(a,b,c,d)$. In $G'$ we build the following  gadget for every double triangle. \mwcom{I think that we may merge two standard gadgets for triangles instead}

Consider any double triangle~$T=(a,b,c,d)$. We remove $c$ and $d$ from $G'$ and we add a vertex $u_T$ to $G'$. We connect $v_T^1$ with $v_T^2$ \mwcom{shouldn't we introduce $v_T^1$ and $v_T^2$ first?} and we connect $u_T$ with both $a$ and $b$. Let $M^i_j(T)$ denote the weight of a maximum weight $C_4$-free $2$-matching of $T$ in which $a$ has degree $i$ and $b$ has degree $j$. We set the weight of edges $(v_T^1,v_T^2)$, $(u_T,a)$ and $(u_T,b)$ to $M_1^1(T)$, $M^2_1(T)-M^1_1(T)$ and $M^1_2(T)-M^1_1(T)$, respectively. We set capacity intervals of $a$, $b$ and $u_T$ to $[0,1]$. We set capacity intervals of $v_T^1$ and $v_T^2$ to $[1,1]$.

\begin{figure}[htpb]
\centering
\begin{tikzpicture}[scale=0.75,transform shape]
	\pic (arrow) [scale=0.3]{transformsTo};
	\pic (dt) [left=4cm of arrow-leftEnd] {doubleTriangle};
	\pic (g) [right=3cm of arrow-rightEnd] {doubleTriangleGadget};
\end{tikzpicture}
\caption{A gadget for a double triangle $T=(a,b,c,d)$.}
\end{figure}
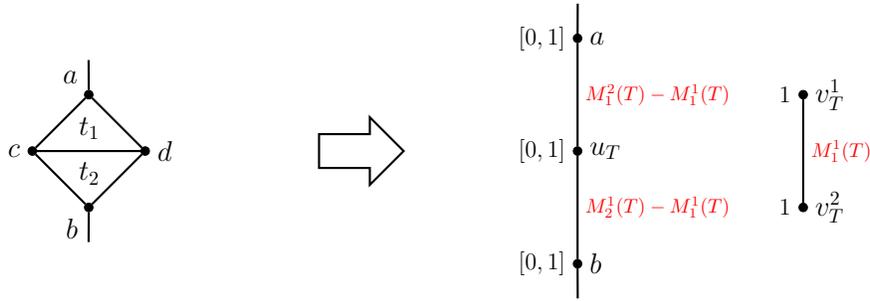

For every problematic short cycle that is not  part of any double triangle we add a corresponding gadget presented in Section~\ref{sec:simple} or Section~\ref{sec:squares}. \mwcom{does it need more explanation?}

\begin{theorem}
Let $M$ be any $C_4$-free $2$-matching of $G$. Then we can find an \lbmatching{} $M'$ of $G'$ such that $w'(M') \geq w(M)$.
\end{theorem}
\begin{proof}
We initialize $M'$ as the empty set. We add to $M'$ every edge of $M$ that belongs to no  problematic short cycle. %nor any double triangle of $G$ to $M'$.

Consider any double triangle $T=(a,b,c,d)$ of $G$. We add $(v_T^1,v_T^2)$ to $M'$. Let $\hat{M}=M\cap E(T)$. If $\deg_{\hat{M}}(a)=2$, then we add $(u_T,a)$ to $M'$. If $\deg_{\hat{M}}(b)=2$, then we add $(u_T,b)$ to $M'$. Note that $\deg_{\hat{M}}(a) \leq 1$ or $\deg_{\hat{M}}(b) \leq 1$, therefore $\deg_{M'}(u_T)\leq 1$.

For every problematic short cycle that is not  part of any double triangle \mwcom{call these cycles somehow} we add edges of a corresponding gadget to $M'$ in the same way as we did in the proofs of Theorem~\ref{lem:opt-to-matching-simple} and Theorem~\ref{lem:square-opt-to-matching}.
\end{proof}

\begin{theorem}
Let $M'$ be any \lbmatching{} of $G'$. Then we can find a $C_4$-free $2$-matching $M$ of $G$ such that $w(M) \geq w'(M')$.
\end{theorem}
\begin{proof}
We initialize $M$ as the empty set. We add to $M$ every edge of $M'$ that belongs to $G$. %For every problematic short cycle of $G$ we add some of its edges to $M$. Next we  replace some edges of $M$ with other ones to remove unproblematic short cycles.

Consider any double triangle $T=(a,b,c,d)$ of $G$. Let $i$ and $j$  denote the number of  edges of the~gadget for  $T$ incident to $a$ and $b$, respectively. Notice   that $i +j \leq 1$. We add to $M$ a maximum weight $C_4$-free $2$-matching of $T$ in which $a$ has degree $i+1$ and $b$ has degree $j+1$. 

For every  short cycle that is not  part of any double triangle we proceed in the same way as in the proofs of Theorem~\ref{thm:single-triangles} and Theorem~\ref{thm:square-matching-to-opt}. \end{proof}

\bibliographystyle{abbrv}
\bibliography{bib}

%\appendix
%\section{Deferred Proofs}
%\label{sec:deferredproofs}

\end{document}